\documentclass[a4paper,UKenglish,cleveref, autoref, thm-restate]{lipics-v2021}
\usepackage{algorithm}
\usepackage{algpseudocode}
\renewcommand{\phi}{\varphi}
\newcommand{\qbf}[1]{$#1$-QBF}
\newcommand{\feqbf}[1]{$\forall\exists\,#1$-QBF}

\usepackage{todonotes}

\title{d-QBF with Few Existential Variables Revisited} 

\author{Andreas Grigorjew}
    {LAMSADE, CNRS UMR7243, Université Paris Dauphine-PSL, 75775 Paris, France}
    {research@andreasgrigorjew.de}
    {https://orcid.org/0000-0003-0989-2415}{}
\author{Michael Lampis}
    {LAMSADE, CNRS UMR7243, Université Paris Dauphine-PSL, 75775 Paris, France}
    {michail.lampis@dauphine.fr}
    {https://orcid.org/0000-0002-5791-0887}{}
\authorrunning{A. Grigorjew and M. Lampis}
\Copyright{Andreas Grigorjew and Michael Lampis}

\category{} 

\relatedversion{} 

\funding{This work was supported by ANR project ANR-21-CE48-0022 (S-EX-AP-PE-AL).}
\nolinenumbers 
\hideLIPIcs

\EventEditors{John Q. Open and Joan R. Access}
\EventNoEds{2}
\EventLongTitle{42nd Conference on Very Important Topics (CVIT 2016)}
\EventShortTitle{CVIT 2016}
\EventAcronym{CVIT}
\EventYear{2016}
\EventDate{December 24--27, 2016}
\EventLocation{Little Whinging, United Kingdom}
\EventLogo{}
\SeriesVolume{42}
\ArticleNo{23}

\begin{document}

\maketitle

\begin{abstract}

Quantified Boolean Formula (QBF) is a notoriously hard generalization of
\textsc{SAT}, especially from the point of view of parameterized complexity,
where the problem remains intractable for most standard parameters. A recent
work by Eriksson et al.~[IJCAI 24] addressed this by considering the case where
the propositional part of the formula is in CNF and we parameterize by the
number $k$ of existentially quantified variables. One of their main results was
that this natural (but so far overlooked) parameter does lead to
fixed-parameter tractability, if we also bound the maximum arity $d$ of the
clauses of the given CNF. Unfortunately, their algorithm has a
\emph{double-exponential} dependence on $k$ ($2^{2^k}$), even when $d$ is an
absolute constant. Since the work of Eriksson et al.\ only complemented this with
a SETH-based lower bound implying that a $2^{O(k)}$ dependence is impossible,
this left a large gap as an open question.

Our main result in this paper is to close this gap by showing that the
double-exponential dependence is optimal, assuming the ETH: even for CNFs of
arity $4$, QBF with $k$ existential variables cannot be solved in time
$2^{2^{o(k)}}|\phi|^{O(1)}$.  Complementing this, we also consider the further
restricted case of QBF with only two quantifier blocks ($\forall\exists$-QBF).
We show that in this case the situation improves dramatically: for each $d\ge
3$ we show an algorithm with running time $k^{O_d(k ^{d-1})}|\phi|^{O(1)}$ and
a lower bound under the ETH showing our algorithm is almost optimal.

\end{abstract}

\section{Introduction}

The Quantified Boolean Formula problem (QBF) is the following: we are given a
formula of the form $Qx_1Qx_2\ldots Qx_n \phi(x_1,x_2,\ldots,x_n)$, where
$x_1,\ldots,x_n$ are boolean variables, each $Q$ is a quantifier ($\exists$ or
$\forall$) and $\phi$ is propositional, and we want to decide if the formula
evaluates to True. QBF is the prototypical PSPACE-complete problem
and clearly generalizes SAT, which is the special case when all quantifiers are
existential. Because of its great generality, QBF captures a large class of
important problems \cite{BeyersdorffJLS21,DBLP:series/faia/336,GiunchigliaMN21,SebastianiT21,ShuklaBPS19}, but at the same time it is considerably more difficult to
solve than SAT, both in theory and in practice \cite{MarinNPTG16}. 

In this paper we focus on the complexity of solving QBF through the lens of
parameterized algorithms\footnote{We assume the reader to be familiar with the
basics of parameterized complexity, as given in standard textbooks
\cite{CyganFKLMPPS15}.}. QBF is notoriously hard from this point of view as
well and most standard parameters which allow tractable algorithms for SAT,
such as treewidth, do not work for QBF (we review some of the relevant
literature below). Nevertheless, a notable recent work of Eriksson, Lagerkvist,
Ordyniak, Osipov, Panolan, and Rychlicki \cite{ErikssonLOOPR24} succeeded in
discovering a rare island of tractability for this problem by considering a
natural but so far overlooked parameter: the number of existentially quantified
variables.  Their main results can be summarized as follows: if a QBF instance
has $\phi$ given in CNF and only contains at most $k$ existential variables,
then the instance can be decided in time roughly $m^{2^k}$, where $m$ is the
number of clauses; while if all clauses of $\phi$ also have arity at most
$d=O(1)$, this running time can be improved to roughly $2^{2^k}|\phi|^{O(1)}$.
The former result is therefore an XP algorithm for parameter $k$ (the number of
existential variables), while the latter is a (rare!) FPT algorithm for QBF
with a natural parameter. 

Even though the results of \cite{ErikssonLOOPR24} are refreshing in the sense
that they uncover a rare natural parameter for which non-trivial algorithmic
results can be obtained for QBF, their \emph{double-exponential} running time
is a major drawback.  This is, however, inevitable for formulas of unbounded
arity, as \cite{ErikssonLOOPR24} supplies a reduction establishing that, under
the Exponential Time Hypothesis (ETH), the $m^{2^k}$ algorithm is optimal.
Nevertheless, in the bounded-arity case the same work only establishes a much
weaker lower bound: no algorithm can solve QBF of arity $3$ in time
$c^k|\phi|^{O(1)}$, for any constant $c$, unless the \emph{Strong} Exponential
Time Hypothesis is false.  This lower bound is quite unsatisfying, as it both
leaves a huge gap in complexity with respect to the upper bound and also relies
on a much stronger (and less widely believed \cite{Williams19}) complexity
hypothesis than the lower bound for CNFs of unbounded arity.

\subparagraph*{Our results} The question of whether QBF with $k$ existential
variables on CNFs of constant arity can be solved faster than with a
double-exponential dependence is stated explicitly in \cite{ErikssonLOOPR24},
where the authors state that it is not clear which direction one should try to
improve. Our main contribution here is to settle this question by proving the
following: assuming the ETH is true, there is no algorithm that can solve QBF
in time $2^{2^{o(k)}}|\phi|^{O(1)}$, even if the propositional part is given in
CNF where all clauses have arity at most $4$. This settles the main question of
\cite{ErikssonLOOPR24} by showing that their algorithm is essentially optimal
and significantly improves their lower bound while relying on a weaker
hypothesis.

Complementing these results we also consider an interesting special case of QBF
where the input formula only has two quantifier blocks, which we denote
$\forall\exists$-QBF. Our motivation for considering this restriction is that
both of the lower bounds given in \cite{ErikssonLOOPR24} already apply to this
special case, so it is interesting to ask if our lower bound also applies.

What we discover is that $\forall\exists$-QBF has a quite different complexity
from general QBF when $d$ is bounded: we show that for all $d\ge 3$ we can
solve $\forall\exists$-QBF in time $k^{O_d\left(k^{d-1}\right)}$, where $O_d$
hides factors depending on $d$, significantly improving upon the
double-exponential dependence of the general case. Furthermore, we show that
our algorithm is essentially optimal, up to a logarithmic factor in the
exponent: for all $d\ge 3$, if $\forall\exists$-QBF could be solved in time
$2^{o(k^{d-1})}|\phi|^{O(1)}$, then the ETH would be false. Observe that this
lower bound also significantly improves the SETH-based lower bound of
\cite{ErikssonLOOPR24}, which proved that $\forall\exists$-QBF of arity $3$
cannot be solved in time $c^k$, as our lower bound in this case is
$2^{o(k^2)}|\phi|^{O(1)}$. These results paint a nice contrast with the case of
unbounded $d$, where \cite{ErikssonLOOPR24} showed a lower bound for
$\forall\exists$-QBF essentially matching their algorithm for general QBF. Our
results imply that for bounded $d$, QBF is strictly harder than
$\forall\exists$-QBF for parameter $k$. 

\subparagraph*{Techniques} For our main (double-exponential) lower bound, our
intuition begins with a reduction given in \cite{LampisM17}, which takes as
input a DNF formula $\phi$ on $n$ variables and $m$ clauses and produces a
$\forall\exists$-QBF instance $\phi'$ on $n+\log m$ variables by adding
$\log m$ new existential variables\footnote{The reduction of \cite{LampisM17}
is actually stated from CNF formulas to $\exists\forall$-QBF, but the
description here is equivalent if we take the complement of all formulas.}. The
reduction ensures the following strong equivalence: for all assignments
$\sigma$ to the common variables of $\phi$ and $\phi'$, we have that $\phi$
evaluates to True for $\sigma$ if and only if $\phi'$ (simplified under
$\sigma$) can be satisfied by some assignment to the $\log m$ new existential
variables. This reduction \emph{almost} works for our purposes, since the
number of existential variables $k$ it produces is $\log m$, so it proves that
any algorithm must have double-exponential complexity in $k$ (otherwise our
reduction would allow us to decide if $\phi$ is valid, which is equivalent to
SAT for $\neg \phi$, in sub-exponential time, refuting the ETH).  The problem
is that the reduction of \cite{LampisM17} produces CNFs with
\emph{non-constant} arity.  More precisely, the reduction represents each term
of $\phi$ with a group of clauses that contain \emph{all} existential
variables, but with a different negation pattern for each term. As a result,
the strategy of the existential player consists of selecting a single such
group of clauses that will be satisfied by their universally quantified
literals, that is, of selecting a term of $\phi$ that evaluates to True.

The main technical obstacle in extending the reduction of \cite{LampisM17} is
therefore to find a way to decrease the arity. We achieve this by adding to the
formula $O(\sqrt{m})$ new universal variables, whose values are supposed to be
fully determined by the values of the $\log m$ existential variables we have
added. In particular, the universal player is ``supposed to'' set exactly two
of these variables to True, and the pair she selects must encode the assignment
to the existential variables. If this happens, it is easy to replace each group
of $\log m$ variables in a clause by a pair of new universal variables.  The
problem, of course, is that the universal player may decide to give values to
these variables which do not agree with what we have in mind. We therefore
construct a new DNF formula $\psi$ which expresses the constraint ``the
universal player cheated on one of the $O(\sqrt{m})$ new variables''. The key
idea now is that we can apply our construction recursively on this new DNF.
Because $\psi$ has size $O(\sqrt{m})$, the number of new existential variables
we use will decrease by a constant factor, hence the total number of
existential variables we introduce until the formula becomes small enough to
brute-force is a geometric series that sums up to $O(\log m)$.  This implies
our double-exponential lower bound under the ETH.


For our algorithmic result for $\forall\exists$-QBF we start with an approach
similar to \cite{ErikssonLOOPR24}, partitioning the clauses of the input
formula into classes which agree on all existential variables (hence
$O_d(k^{d-1})$ classes, for clauses containing at least one universal
variable).  However, whereas \cite{ErikssonLOOPR24} attempts to find sunflowers
in such classes, we attempt to find a large collection of clauses that share no
universal variable. Our motivation is that, unlike \cite{ErikssonLOOPR24}, we
do not attempt to use such a structure to remove a single clause, but rather to
argue (via the probabilistic method) that if all classes have such large
collections then the universal player has an excellent strategy: she can assign
values to her variables so that for all classes, the existential literals
remain in one of the clauses. If we are unable to find a large enough
collection in one class, this gives us a small hitting set of the clauses of
this class, on which we perform a simple branching algorithm.

Finally, for the lower bounds that prove that the above algorithm is almost
optimal, we again rely on the strategy of the reduction of \cite{LampisM17},
but this time instead of introducing $\log m$ existential variables we
introduce $m^{\frac{1}{d-1}}$ such variables. The idea is that a group of $d-1$
such variables will identify a term of the original formula, and together with
a literal of that term we obtain a clause of arity $d$.

\subparagraph*{Related work} As mentioned, QBF is  much harder than SAT from
the point of view of parameterized complexity. Notably, QBF is PSPACE-complete
for formulas of constant pathwidth \cite{AtseriasO14}. If we restrict the
number of quantifier alternations the problem does become FPT parameterized by
treewidth \cite{Chen04}, but the parameter dependence is a tower of
exponentials of height equal to the number of alternations
\cite{FichteHP20,LampisMM18,LampisM17,PanV06}. This has motivated the study of
QBF under extremely restricted parameters such as primal vertex cover, for
which the problem does become FPT \cite{EibenGO20}, albeit with a
double-exponential dependence on the parameter which is already optimal with
one quantifier alternation \cite{LampisM17}. It has also motivated the study of
structural parameters which take into account the quantification ordering
\cite{EibenGO20}, backdoor-type parameters \cite{SamerS09}, and parameters that
fall between vertex cover and treewidth \cite{FichteGHSO23}.

\section{Preliminaries}
\label{sec:prelim}

We use standard notation and recall some basic notions. A \qbf{d} is a formula
of the form $\phi = (Q_1x_1)\dots(Q_nx_n)\psi$, where $Q_i \in \{ \exists,
\forall \}$ and $\psi$ is in $d$-CNF. Recall that a propositional boolean
formula is in $d$-CNF if it is a conjunction of clauses of literals, where each
clause contains at most $d$ literals. Conversely, a formula is in $d$-DNF if it
is a disjunction of terms, where each term is a conjunction of at most $d$
literals. We will sometimes view clauses (or terms) as sets of literals, and
CNF (or DNF) formulas as sets of clauses (or terms respectively). By \feqbf{d}
we denote a \qbf{d} $\phi$ where $Q_1,\dots,Q_{n-k} = \forall$ and $Q_{n-k+1},
\dots,Q_n = \exists$. We write shorthand $\phi = \forall \mathbf{y} \, \exists
\mathbf{x} \psi$, where the notation $\mathbf{x}$ is meant to indicate that
$\mathbf{x}$ is a tuple (or vector) of boolean variables. For an assignment
$\sigma$ to (some of) the variables of $\phi$, we denote by $\phi(\sigma)$ the
formula obtained by applying the assignment to $\phi$: if a literal evaluates
to False, it gets removed from the clause, and if a literal evaluates to True,
we remove the clause from the formula. We also write $\phi(\mathbf{x})$ to
indicate that $\phi$ is a formula with variables in $\mathbf{x}$. We will
sometimes view QBF as a game between two players, the universal and the
existential player. At turn $i$, the appropriate player according to $Q_i$
assigns a value to the variable $x_i$, and the existential player's goal is to
arrive at an assignment that evaluates to True.  It is not hard to see that the
existential player has a winning strategy if and only if the quantified formula
evaluates to True.

The Exponential Time Hypothesis (ETH) states that there exists $c>1$ such that
3-SAT cannot be solved in time $c^n$ \cite{ImpagliazzoPZ01}, while its strong
variant (SETH) states that for all $\varepsilon>0$ there exists $k$ such that
$k$-SAT cannot be solved in time $(2-\varepsilon)^n$ \cite{ImpagliazzoP01}. All
our lower bounds in this paper rely on the ETH (which is considered more solid
than the SETH). To ease notation, we will use the slightly more informal
variation of the ETH stating that 3-SAT cannot be solved in time $2^{o(n)}$.
Furthermore, the sparsification lemma of \cite{ImpagliazzoPZ01} implies that if
the ETH is true, then 3-SAT cannot be solved in time $2^{o(n)}$, even for
instances with $m=O(n)$ clauses.

\section{Lower Bounds}
\label{sec:lower-bounds}

In this section we present two lower bounds based on the ETH. The first
(Theorem~\ref{thm:doubleLB}) shows that the double-exponential algorithm of
\cite{ErikssonLOOPR24} is essentially optimal, even if we restrict ourselves to
CNFs of arity $4$. The second (Theorem~\ref{thm:singleLB}) shows that the
algorithm we present in Section~\ref{sec:fe-qbf} is also essentially optimal, up
to poly-logarithmic factors in the exponent.

\subsection{Double-Exponential Lower Bound}

Before we proceed to the proof, let us start with a useful definition. We
define a function $B$, which takes as arguments a tuple of $n$ boolean
variables $(x_1,\ldots,x_n)$ and an integer $i\in\{0,\ldots,2^n-1\}$ and
outputs a clause on the given variables with a negation pattern that represents
the integer $i$. In particular, we set $B(0, \emptyset) = \emptyset$ as the
base case and $B(i, (x_1,\dots,x_n)) = B(\lfloor i/2 \rfloor,
(x_1,\dots,x_{n-1})) \lor  l_n $, where $l_n = x_n$ if $i$ is even, and
$l_n = \neg x_n$ otherwise. Observe that $B$ is defined in such a way that the
clause $B(i, (x_1,\dots,x_n))$ is falsified if and only if the assignment we
give to the variables corresponds to the binary representation of the integer
$i$.

\begin{theorem}\label{thm:doubleLB} There is no algorithm which, given a
\qbf{4} instance $\phi$ with n variables and $k$ existential variables, can
decide if $\phi$ is true in time $2^{2^{o(k)}}\cdot n^{O(1)}$, unless the ETH
is false.  \end{theorem}

\begin{proof}
    
We construct a polynomial-time algorithm that takes as input a DNF
$\psi(\mathbf{x})$ with $n$ variables and $m$ clauses, and outputs a
$4$-CNF $\phi(\mathbf{x}, \mathbf{y})$ and a tuple $Q \in \{
\exists, \forall \}^{|\mathbf{y}|}$, such that the following two conditions are
satisfied:

    \begin{enumerate}
        \item Equivalent formulas: for all assignments $\sigma: \mathbf{x} \to \{ 0, 1 \}$ we have  $\psi(\sigma) \iff Q(\mathbf{y}) \phi(\sigma)$, \label{cond:equivalent}
        \item Bounded number of existential variables in $\mathbf{y}$: $|\{ q_i \in Q \mid q_i = \exists \}| \leq 6\log_2(n+m) + O(1)$. \label{cond:few-existential}
    \end{enumerate}

Before we proceed, let us explain why such an algorithm establishes the
theorem. Suppose we take as input a $3$-DNF formula $\psi$ on $n$ variables and
we want to decide if $\psi$ is valid. One way to do this would be to give it as
input to our algorithm and produce an equivalent 4-QBF instance, where the $n$
variables of $\psi$ are first in the ordering and universally quantified. The
algorithm would run in time $2^{2^{o(k)}} \cdot n^{O(1)}$, and since $k=O(\log
n)$ the running time would be $2^{o(n)}$. However, such an algorithm is
impossible under the ETH.

Let us then prove that such an algorithm exists.  We use induction on $n+m$, or
more precisely, we describe a recursive algorithm which has a base case when
$n+m$ is small (say $n+m<2^{60}$, where we just pick this constant for
concreteness and ease of presentation without trying to optimize) and otherwise
performs a construction which calls the same algorithm on a smaller DNF.  

For the base case, assume $n+m \leq 2^{60}$.  Since $\psi$ is ``small'' we can
convert into a 4-CNF by brute force. Namely, we construct a clause per
falsifying assignment as follows.  Enumerate all the $2^n$ assignments
$\sigma:\mathbf{x} \to \{ 0, 1 \}$ and, if $\phi(\sigma)$ is False, encode it
as CNF formula:

    \[ \mathcal{C}_\mathbf{x} = \left\{ (l_1 \vee \dots \vee l_n ) \mid \exists \sigma:\mathbf{x} \to \{ 0, 1 \} : \text{ $\phi(\sigma)$ is False, and }
        l_i =
        \begin{cases}
            x_i,      &\text{ if } \sigma(x_i) = 0,\\
            \neg x_i,  &\text{ else.}
        \end{cases}
    \right\} \]  

In other words, for each falsifying assignment of $\psi$ we construct a clause
all of whose literals evaluate to False if we use this assignment.  Let
$\phi(\mathbf{x}) = \bigwedge_{c \in \mathcal{C}_\mathbf{x}} c$.  Condition
\ref{cond:equivalent} is satisfied, since $\psi(\sigma)$ is true if and only if
$\phi(\sigma)$ is true, but the new formula is not in $4$-CNF yet, as it
contains at most $2^n$ clauses, each of arity $n$. We introduce at most $n2^n$
new existential variables which allow us to break down the clauses into clauses
of arity $4$ -- in particular there exists an assignment to the new variables
satisfying the formula if and only if the original formula was True.  Since
$n<2^{60}$, the new formula has $O(1)$ existential variables, so now all
conditions are satisfied.

Assume now $n+m>2^{60}$ and assume that the algorithm works correctly for all
pairs $n',m' \in \mathbb{N}$ with $n' + m' < n + m$.  We further assume that $m
= 4^\ell$, as we can otherwise repeat clauses, increasing $m$ at most by a
factor of $4$ (we explain below why this is not a problem).

Suppose the terms of $\psi(\mathbf{x})$ are numbered $\mathcal{T} = \{ T_0,
\dots, T_{m-1} \}$.  Lampis and Mitsou~\cite{LampisM17} defined an algorithm that,
given $\psi(\mathbf{x})$, constructs an equivalent QBF formula
$\psi'(\mathbf{x}, \mathbf{y})$ in CNF using a binary encoding.  In particular,
we construct $\psi'$ by replacing every term $T_i = (l_1 \wedge \dots \wedge
l_d)$ by the $d$ clauses $\mathcal{C}_i=(l_1\lor B(i,\mathbf{y}))\land (l_2\lor
B(i,\mathbf{y}))\land\ldots \land (l_d\lor B(i,\mathbf{y}))$.  The formula
$\psi'(\mathbf{x},\mathbf{y})$ satisfies our conditions if all variables
$\mathbf{y}$ are existentially quantified (this is not hard to see and proved
in detail in~\cite{LampisM17}).  However, the arity of $\psi'$ is at least
$\log_2m + 1>4$. We will modify the construction to shrink the arity.  We
introduce the following variables, quantified in this order: 

    \begin{itemize}
        \item $\log_2m$ existential variables $\mathbf{y} = \{ y_1, \dots, y_{\log_2m} \}$ (these are the same as in \cite{LampisM17}),
        \item $2\sqrt{m}$ universal variables $\mathbf{z}^1 = \{ z^1_1, \dots, z^1_{\sqrt{m}} \}, \mathbf{z}^2 = \{ z^2_1, \dots, z^2_{\sqrt{m}} \}$,
        \item one existential variable $w$.
    \end{itemize}

Consider a bijection $\lambda:\{0,\ldots,m-1\} \to \{0,\ldots,\sqrt{m}-1\}
\times \{0,\ldots,\sqrt{m}-1\}$. In particular, let us use the natural
bijection $\lambda(i)=(\lfloor \frac{i}{\sqrt{m}}\rfloor, i\bmod\sqrt{m})$. We
write $\lambda(i)_1, \lambda(i)_2$ to refer to the two integers returned by
$\lambda$. We define a QBF $\phi(\mathbf{x}, \mathbf{y},
\mathbf{z}^1, \mathbf{z}^2, w)$ that consists of the following two parts.  In
the first part, which will be in $4$-CNF, for every term $T_i = (l_1 \wedge
\dots \wedge l_d)$, we introduce $d$ many clauses $\mathcal{C}_i(\mathbf{x},
\mathbf{z}^1, \mathbf{z}^2, w) = (l_1 \vee \neg z^1_{\lambda(i)_1} \vee \neg
z^2_{\lambda(i)_2} \vee w) \wedge \dots \wedge (l_d \vee \neg
z^1_{\lambda(i)_1} \vee \neg z^2_{\lambda(i)_2} \vee w)$.  The idea here is
that the variable vector $\mathbf{z}^1$ represents $y_1,\dots,y_{\log_2m/2}$,
and $\mathbf{z}^2$ represents $y_{\log_2m/2+1},\dots,y_{\log_2m}$ in the sense
that the universal player is ``supposed to'' set both
$z^1_{\lambda(i)_1},z^2_{\lambda(i)_2}$ to True only if the clause
$B(i,\mathbf{y})$ evaluates to False, that is, if the assignment to
$\mathbf{y}$ is the binary representation of $i$.  Under this assumption, we
are using a pair of universal variables in each clause to replace the $\log m$
existential variables in the construction of~\cite{LampisM17}. We also add the
new existential variable $w$ everywhere to allow the existential player to have
a response in case the universal player does not play in the way described
above.

In the second part, where we check exactly if the universal player played as
expected, we define a DNF formula by:

    \[ \mathcal{D}'(\mathbf{y}, \mathbf{z}^1, \mathbf{z}^2, w) = (\neg w) \vee \bigvee_{i \in [\log_2m/2], j \in \{0,\ldots,\sqrt{m}-1\}} (z^1_j \wedge l_{i,j}) \bigvee_{i \in [\log_2m/2], j \in \{0,\ldots,\sqrt{m}-1\}} (z^2_j \wedge l'_{i,j}), \]
    where we set $\{ l_{1,j}, \dots, l_{\log_2m/2,j} \} = B(j, (y_1,\ldots,y_{log_2m/2}))$ and also $\{ l'_{1,j}, \dots, l'_{\log_2m/2,j} \} = B(j, (y_{\log_2m/2+1},\ldots,y_{log_2m}))$.

The idea here is the following: if $z^1_{j_1}$ and $z^2_{j_2}$ are both set to
True, we consider the clause $B({\sqrt{m}}\cdot j_1+j_2,\mathbf{y})$, where
${\sqrt{m}}\cdot j_1+j_2$ is the number we obtain by concatenating the binary
representations of $j_1$ and $j_2$. If this clause contains a True literal,
then $\mathcal{D}'$ contains a True term other than $\neg w$. Therefore, if the
universal player wants to avoid setting a term of $\mathcal{D}'$ to True, she
can set at most one pair of $z^1_{j_1},z^2_{j_2}$ variables to True, namely,
those variables such that $B({\sqrt{m}}\cdot j_1+j_2,\mathbf{y})$ evaluates to
False.

Using both parts, we define
    \[ \phi'(\mathbf{x}, \mathbf{y}, \mathbf{z}^1, \mathbf{z}^2, w) = \bigwedge_{i \in [m]} \mathcal{C}_i(\mathbf{x}, \mathbf{z}^1, \mathbf{z}^2, w) \wedge \mathcal{D}'(\mathbf{y}, \mathbf{z}^1, \mathbf{z}^2, w). \]

Note that $\phi'$ is not in CNF yet. Before reducing it to such, however, we show that $\phi'$ fulfils Condition \ref{cond:equivalent}.
    We show that for all $\sigma : \mathbf{x} \to \{ 0, 1 \}$ we have $ \psi(\sigma) \iff \exists \mathbf{y} \forall \mathbf{z}^1 \forall \mathbf{z}^2 \exists w \phi'(\mathbf{x}, \mathbf{y}, \mathbf{z}^1, \mathbf{z}^2, w)$.

Let $\psi(\sigma)$ be satisfied.  That is, there exists $T_i = (l_1 \wedge
\dots \wedge l_d) \in \psi$ such that all literals of $T_i(\sigma)$ are true.
Choose $\sigma':\mathbf{y} \to \{ 0, 1 \}$, $\sigma'(y_j) = 1$ if $\neg y_j \in
B(i, \mathbf{y})$, else $\sigma'(y_j) = 0$. In other words, we choose $\sigma'$
so that $B(i,\mathbf{y})$ is False and for all $i'\neq i$ we have
$B(i',\mathbf{y})$ is True. We analyze the outcomes of all possible choices for
$\mathbf{z}^1$ and $\mathbf{z}^2$.

\begin{itemize}

\item If $z^1_j = 0$ for all $j \in \{0,\ldots,\sqrt{m}-1\}$ or $z^2_j = 0$ for
all $j \in \{0,\ldots,\sqrt{m}-1\}$, then every clause in $\mathcal{C}_j$ is
satisfied for all $j \in \{0,\ldots,m-1\}$. The existential player then chooses
$w = 0$ and $\mathcal{D}$ is satisfied.

\item If $z^1_{j_1} = z^2_{j_2} = 1$ for some $j_1, j_2 \in
\{0,\ldots,\sqrt{m}-1\}$, such that $\lambda(i) \neq (j_1, j_2)$, then a term
of $\mathcal{D}'$ other than $(\neg w)$ is True.  In particular, if
$\lambda(i)_1 \neq j_1$, the term $(z^1_{j_1} \wedge l_{t, j_1})$ is true for
some $t \in [\log_2{m}/2]$ (analogously for $\lambda(i)_2 \neq j_2$). The
existential player can thus choose $w = 1$ to satisfy all $\mathcal{C}_j$.

\item If $z^1_{j_1} = z^2_{j_2} = 1$ only for a specific pair $j_1, j_2 \in
\{0,\ldots,\sqrt{m}-1\}$, such that $\lambda(i) = (j_1, j_2)$ (and all other
$z$ variables are False), then $\mathcal{C}_i$ is true, because all literals of
$T_i$ are true.  For all $i' \neq i$, $\mathcal{C}_{i'}$ is True, because the
other $z$ variables are False.  If the existential player sets $w = 0$, we also
satisfy $\mathcal{D}$.

    \end{itemize}

Thus, $\exists \mathbf{y} \forall \mathbf{z}^1 \forall \mathbf{z}^2 \exists w
\phi'(\sigma)$ is true and we have established one direction.
    
Let now $\exists \mathbf{y} \forall \mathbf{z}^1 \forall \mathbf{z}^2 \exists w
\phi'(\sigma)$ be true.  For the given $\mathbf{y}$, consider $i =
\sum_{j=1}^{\log_2m} y_j 2^{j-1}$.  We further choose $z^1_{j_1} = z^2_{j_2} =
1$ for the pair $j_1, j_2 \in \{0,\ldots,\sqrt{m}-1\}$ such that $\lambda(i) =
(j_1, j_2)$, (in particular $j_1=\lfloor\frac{i}{\sqrt{m}}\rfloor$ and
$j_2=i\bmod \sqrt{m}$) and we choose all other $z$ variables to be $0$.  Since
$\phi'(\sigma)$ is True, it must be True for this assignment to $\mathbf{z}^1,
\mathbf{z}^2$. But for this assignment, the existential player is forced to set
$w$ to False, as otherwise $\mathcal{D}'$ has no True term. Therefore,
$\mathcal{C}_i$ must be true for this choice of $\mathbf{z}^1, \mathbf{z}^2$,
so all the literals of the term $T_i$ must be True.   Hence, $\phi(\sigma)$ is
satisfied.

The formula $\phi'$ now consists of a $4$-CNF part
$\mathcal{C}_i(\mathbf{x}, \mathbf{z}^1, \mathbf{z}^2, w)$ ($i\in\{
0,\dots,m-1\}$) and of a DNF part $\mathcal{D}'(\mathbf{y},
\mathbf{z}^1, \mathbf{z}^2, w)$.  Note that $\mathcal{D}'$ contains at most $n'
\leq 2 \sqrt{m} + \log_2m + 1$ many variables and $m' \leq 2\sqrt{m}\log_2m+1$
many terms.  If we had initially increased $m$ to make it a power of $4$, these
bounds become $n'\le 4\sqrt{m}+\log_2m+3$ and $m'\leq
4\sqrt{m}\log_2m+4\sqrt{m}+1$. We therefore have $n'+m'\leq 4\sqrt{m}\log_2m +
8\sqrt{m} + log_2m+3 \le 16 \sqrt{m+n}\log_2(m+n)$. However, when $m+n>2^{60}$
we have $16\sqrt{m+n}\log_2(m+n)<(m+n)^{2/3}$.  Thus, using induction (that is,
a recursive call to our algorithm) on $\mathcal{D}'$, we can replace it by an
equivalent QBF in $4$-CNF $\mathcal{D}(\mathbf{y},
\mathbf{z}^1, \mathbf{z}^2, w, \mathbf{y}')$, using at most
$6\log_2((m+n)^{2/3}) + O(1) = 4\log_2(m+n) + O(1)$ many existential variables.
We finally define

    \[ \phi(\mathbf{x}, \mathbf{y}, \mathbf{z}^1, \mathbf{z}^2, w, \mathbf{y}') = \bigwedge_{i \in [m]} \mathcal{C}_i(\mathbf{x}, \mathbf{z}^1, \mathbf{z}^2, w) \wedge \mathcal{D}(\mathbf{y}, \mathbf{z}^1, \mathbf{z}^2, w, \mathbf{y}')  \]

The formula $\phi$ is a $4$-QBF-CNF and uses at most
$\log_2m + (4 \log_2(m+n) + O(1)) \leq 6\log_2(m+n) + O(1)$ many additional
existential variables.  Hence, both Conditions \ref{cond:equivalent} and
\ref{cond:few-existential} are fulfilled.  \end{proof}

\subsection{Almost-Tight Lower Bound for Two Quantifier Blocks}

Analogously to the previous section, we define a helpful function that
translates integers to clauses. Because we will construct a family of
reductions, one for each $d\ge 3$, we will need a family of such functions. Let
$t\ge 1$ be an integer, $\mathbf{x}^1$, $\mathbf{x}^2,\ldots,\mathbf{x}^t$ be
$t$ tuples of $n$ boolean variables each, numbered $x^j_0,\ldots,x^j_{n-1}$ for
$j\in[t]$, and $i$ an integer in $\{0,\ldots,n^t-1\}$. We define
$B_1(i,\mathbf{x}^1)=\neg x^1_i$ and
$B_t(i,\mathbf{x}^1,\ldots,\mathbf{x}^t)=B_{t-1}(\lfloor
\frac{i}{n}\rfloor,\mathbf{x}^1,\ldots,\mathbf{x}^{t-1})\lor \neg x^t_{i\bmod
n}$. Another way to see this function is that if we want to express the integer
$i<n^t$ with a clause using variables from $t$ tuples of $n$ variables, we
write $i$ in base $n$, and each digit gives us the index of a variable of a
tuple, which we then add negated to the clause. Intuitively this function will
be useful, as it allows us, as we increase the arity, to use fewer variables to
express the same integers. 

\begin{theorem}\label{thm:singleLB}
    For all $d \geq 3$, there exists no algorithm solving \feqbf{d} with $k$ existential variables in time $2^{o(k^{d-1})}\cdot n^{O(1)}$, unless the ETH is false.
\end{theorem}

\begin{proof}

We describe a polynomial-time algorithm that takes as input a
DNF $\psi(\mathbf{x})$ with $n$ variables and $m=O(n)$ terms,
and produce a $d$-CNF $\phi(\mathbf{x}, \mathbf{y})$, such that
the following two conditions are satisfied:

\begin{enumerate}
   \item Equivalent formulas: for all assignments $ \sigma:\mathbf{x} \to \{0, 1\}$ we have $ \psi(\sigma) \iff \exists \mathbf{y} \phi(\sigma)$, \label{cond:equivalent-1}
   \item The number of existential variables in $\mathbf{y}$ is upper bounded by $O(m^{1/(d-1)})$. \label{cond:few-existential-1}
\end{enumerate}

Before we proceed, let us explain why such an algorithm establishes the
theorem. Suppose we are given a $3$-DNF formula of $n$ variables and
$m = O(n)$ clauses $\psi$ and want to decide if $\psi$ is valid.  Using the
reduction above, there exists an equivalent \feqbf{d} $\phi$ with
$O(n^{1/(d-1)})$ many existential variables.  Thus, an algorithm in time
$2^{o(k^{d-1})}$ for \feqbf{d} can decide if a given $3$-DNF is
valid in time $2^{o(n)}$, which contradicts the ETH.

Without loss of generality, we can assume that the number of terms $m$ of
$\psi$ satisfies that $m^{\frac{1}{d-1}}$ is an integer, as we can achieve this
by adding repeated clauses at most $O(n)$ times.  We introduce $(d-1) \cdot
m^{1/(d-1)}$ existential variables $\mathbf{y}^1, \dots \mathbf{y}^{d-1}$,
where each $\mathbf{y}^i$ is of size $m^{1/(d-1)}$.  Number the terms in $\psi$
by $\mathcal{T} = \{ T_0,\dots, T_{m-1} \}$.  Consider the natural bijection
$\lambda: \{0,\ldots,m-1\} \to \{0,\ldots,m^{1/(d-1)-1}\}^{(d-1)}$ defined for
each $i$ by writing $i$ in base $m^{1/(d-1)}$.  For every term $T_i = (l_{j_1}
\wedge \dots \wedge l_{j_{d'}})$  we define the clauses
$\mathcal{C}_i(\mathbf{x}, \mathbf{y}^1, \dots, \mathbf{y}^{d-1}) = (l_{j_1}
\vee B_{d-1}(i,\mathbf{y}^1,\ldots,\mathbf{y}^{d-1})) \wedge \dots \wedge
(l_{j_{d'}}\vee B_{d-1}(i,\mathbf{y}^1,\ldots,\mathbf{y}^{d-1}) )$. Observe
that these clauses all have arity $d$.

Furthermore, we also add the clauses
$\mathcal{D}(\mathbf{y}^1,\dots,\mathbf{y}^{d-1}) = (y^1_1 \vee \dots \vee
y^1_{m^{1/(d-1)}}) \wedge \dots \wedge (y^{d-1}_1 \vee \dots \vee
y^{d-1}_{m^{1/(d-1)}})$. Observe that these clauses have arity $m^{1/(d-1)}$,
which is too high, but we will fix this in the end.

    We define
    \[ \phi(\mathbf{x}, \mathbf{y}^1, \dots, \mathbf{y}^{d-1}) = \bigwedge_{j \in \{0,\ldots,m^{1/(d-1)}-1\}} \mathcal{C}_{j}(\mathbf{x}, \mathbf{y}^1, \dots, \mathbf{y}^{d-1}) \wedge \mathcal{D}(\mathbf{y}^1, \dots, \mathbf{y}^{d-1}). \]

Note that $\phi$ fulfils Condition \ref{cond:few-existential-1}, since it uses
$(d-1) \cdot m^{1/(d-1)} = O(m^{1/(d-1)})$ many variables.  We now show that
$\phi$ fulfils Condition \ref{cond:equivalent-1}.  Let $\psi(\sigma)$ be True,
that is, a term $T_i \in \mathcal{C}$ is True. Then we can select
$y^1_{\lambda(i)_1} = \dots = y^{d-1}_{\lambda(i)_{d-1}} = 1$, and set all
other $y$ variables to $0$.  Then every clause in $\mathcal{D}$ is satisfied.
Furthermore, every clause in $\mathcal{C}_j$ with $j \neq i$ is satisfied,
because $B_{d-1}(j,\mathbf{y}^1,\ldots,\mathbf{y}^{d-1})$ is satisfied by our
assignment for all $j\neq i$.  Further, $\mathcal{C}_i$ is satisfied, because
$T_i$ is satisfied.  Thus, $\phi(\sigma)$ is true.

Let $\phi(\sigma)$ be true.  Since $\mathcal{D}(\mathbf{y}^1, \dots,
\mathbf{y}^{d-1})$ is true, at least one variable in each $\mathbf{y}^i$ is
True.  Let $y^1_{i_1}, \dots, y^{d-1}_{i_{d-1}}$ be such variables and let $i$
be the number obtained if we concatenate $i_1,\ldots,i_{d-1}$ and interpret
them as a number in base $m^{1/(d-1)}$. We observe that the clause
$B_{d-1}(i,\mathbf{y}^1,\ldots,\mathbf{y}^{d-1})$ is falsified by $\sigma$, so
if $\phi$ is satisfied, the clauses of $\mathcal{C}_i$ are True because of the
literals of $T_i$. Hence, a term of the original DNF $\psi$ is True for
$\sigma$.

Finally, let us explain how to fix the fact that the clauses of $\mathcal{D}$
have high arity. We have $d$ clauses of arity $m^{1/(d-1)}$. We can add fresh
existential variables and use the standard trick that replaces a long clause
$(l_1\lor l_2\lor l_3\lor\ldots\lor l_t)$ with two clauses $(l_1\lor l_2\lor
z)\land (\neg z\lor l_3\lor\ldots\lor l_t)$. Applying this repeatedly will
require a number of fresh variables bounded by $dm^{1/(d-1)}$, so the total
number of existential variables remains $O(m^{1/(d-1)})$.  \end{proof}

\section{Almost tight algorithm for \feqbf{d}}
\label{sec:fe-qbf}

In this section, we will show the following Theorem.

\begin{theorem}\label{thm:algo}
    There exists an algorithm that, given a \feqbf{d} $\phi = \forall\mathbf{y}
    \exists\mathbf{x}\, \psi$, verifies whether $\phi$ is True or False in time
    $O\left(k^{O_d(k^{d-1})}\right) \cdot |\phi|^{O(1)}$, where $k =
    |\mathbf{x}|$ is the number of existential variables.
\end{theorem}

By \Cref{thm:singleLB}, this runtime is almost tight, (i.e., tight up to a
logarithmic factor of $k$ in the exponent). We provide a recursive algorithm
and prove its correctness and time complexity by induction.

Our algorithm begins as follows: Given the $d$-CNF $\psi$, we partition the
clauses into groups according to their existential literals. More precisely
define for all sets $C \subseteq \{ x, \neg x \mid x \in \mathbf{x} \}$ the
following two sets.

\begin{align*}
    G^\psi_C &= \left\{ c \in \psi \mid c \cap \{ x, \neg x \mid x \in
        \mathbf{x} \} = C \vphantom{G^\psi_C} \right\},\\
    \mathcal{S}^{\psi}_C &= \left\{ c\setminus C \mid c\in G^{\psi}_C \right\}.
\end{align*}

We call the sets $G^\psi_C$ groups of clauses. They define a partition
$\mathcal{P}_\psi$ of the clauses, that is, $G^\psi_C$ contains exactly the
clauses whose existential part is the clause $C$.  The sets
$\mathcal{S}^{\psi}_C$ consist of the clauses in $G^\psi_C$ restricted to their
universal variables. We call $C$ the \textit{existential core} (or simply core)
of the clause $c \in \psi$ if $c \in G_C$.

In the algorithm, we utilize a ``size threshold'' $X = X(k, d) = 2^d d \ln k$,
whose choice of value will be explained in the proofs below. Our algorithm is
recursive. The base case is defined as the following situation: for every core
$C$, the set family $\mathcal{S}^{\psi}_C$ either consists of the empty clause
(this happens if $\psi$ has a clause with only existential variables), or
contains a large number of pairwise disjoint clauses (i.e., at least $X$
clauses that share no universal variables).  The first step of our algorithm is
to greedily try to construct a maximal set of pairwise disjoint clauses in each
$\mathcal{S}^{\psi}_C$. If this succeeds, that is we find such sets of size $X$
for all $\mathcal{S}^{\psi}_C$, then we are in the base case.  We then remove
from all clauses all universal variables (that is, we replace each clause with
its core) and return True if and only if the resulting
(existentially-quantified) formula is satisfiable.  This base case is then
solved in time $2^k |\psi|^{O(1)}$. We will present an argument based on the
probabilistic method justifying why this reduction is correct, that is, the
universal player has a strategy to reduce the formula to its existential core,
which is clearly optimal for her.

If the greedy construction of disjoint sets does not succeed for some
$\mathcal{S}^{\psi}_C$, we are in the  recursion step, and observe that we have
a small hitting set $H$, that is, a set of universal variables of size at most
$dX$ that intersects all clauses of $\mathcal{S}^{\psi}_C$ (this is formed
simply by the union of the maximal set of disjoint clauses).  In this case, we
consider all possible assignments $\sigma: H \to \{ 0, 1 \}$, and run the
algorithm on all formulas $\psi(\sigma)$. If one of the recursive calls returns
that the formula is False, we return False. The key observation is that,
because we have a hitting set, in each recursive call we have reduced the arity
of all clauses of $\mathcal{S}^{\psi}_C$, and this allows us to bound the
recursion depth.

Throughout we will assume that $\psi$ contains no clauses of only universal
variables, since such a clause is either a tautology and can safely be removed,
or the universal player can find a winning strategy and falsify the formula.
It follows that all cores are non-empty.

\subsection{Base Case and Probabilistic Argument}

As mentioned the first step of our algorithm is to greedily construct a maximal
set of clauses for each $\mathcal{S}^{\psi}_C$ which share no universal
variables.  For the base case, suppose this step succeeds in constructing, for
each set $\mathcal{S}^{\psi}_C \neq \{ \emptyset \}$, such a set of at least
$X$ clauses sharing only their core.

Using the probabilistic method, we provide an optimal strategy for the
universal player. Given the formula $\psi(\mathbf{y}, \mathbf{x})$, we
construct a new formula $\psi|_\mathbf{x}(\mathbf{x})$ by replacing each clause
of $\psi$ by a new clause:
\[ \psi \ni (l_{\mathbf{x}, 1} \vee \dots \vee l_{\mathbf{x},i} \vee
    l_{\mathbf{y}, 1} \vee \dots \vee l_{\mathbf{y}, j}) \mapsto
    (l_{\mathbf{x}, 1} \vee \dots \vee l_{\mathbf{x}, i}) \in
    \psi|_\mathbf{x}, \]
where $l_{\mathbf{x}, 1}, \dots, l_{\mathbf{x}, i} \in \{x, \neg x \mid x \in
\mathbf{x}\}$ and $l_{\mathbf{y}, 1}, \dots, l_{\mathbf{y}, j} \in \{y, \neg y
\mid y \in \mathbf{y}\}$. That is, we reduce all clauses of $\psi$ to their
core literals (and then of course retain only one copy of each resulting clause).

\begin{lemma}\label{lem:reduce}

Let $\psi(\mathbf{y}, \mathbf{x})$ satisfy the property that for all
$\mathcal{S}^{\psi}_C\neq\{\emptyset\}$ there exists a set of at least $X$
disjoint clauses in $\mathcal{S}^{\psi}_C$ sharing no universal variables.
Then, we have \[ \forall\mathbf{y}\exists \mathbf{x} \psi \iff \exists
\mathbf{x} \psi|_\mathbf{x}.  \]

\end{lemma}

\begin{proof}

We observe that one direction is trivial, that is $\exists \mathbf{x}
\psi|_\mathbf{x}\Rightarrow \forall\mathbf{y}\exists \mathbf{x} \psi$, because
if the existential player can satisfy the cores of all clauses she can clearly
satisfy the original formula. Therefore, we want to show that $\neg\exists
\mathbf{x} \psi|_\mathbf{x}\Rightarrow \neg\forall\mathbf{y}\exists \mathbf{x}
\psi$ or equivalently that $\forall \mathbf{x} \neg\psi|_\mathbf{x}\Rightarrow
\exists\mathbf{y}\forall \mathbf{x} \neg\psi$. Suppose then that
$\psi|_\mathbf{x}$ is not satisfiable and we want to find a winning strategy
for the universal player. We show that there is a strategy whose outcome is to
simplify the formula into $\psi|_\mathbf{x}$ by showing that if we pick a
random assignment this outcome has positive probability.

    Choose an assignment $\sigma: \mathbf{y} \to \{ 0, 1 \}$ uniformly at
    random. That is, every $y \in \mathbf{y}$ gets assigned to True
    independently with probability 1/2.
    Let $C \subseteq \{x. \neg x \mid x \in \mathbf{x} \}$ and let
    $\mathcal{S}^{\psi}_C\neq\{\emptyset\}$. By assumption, the family
    $\mathcal{S}^{\psi}_C$ contains at least $X$ pairwise disjoint sets. That
    is, the group $G^{\psi}_C$ contains $X$ clauses $C_1, \dots, C_{\lceil
    X\rceil}$ whose universal variables are pairwise disjoint.

    Given an assignment $\sigma$, we call $c \in \psi(\sigma)$
    \textit{active}, if the assignment does not satisfy $c$. We call
    the clause \textit{inactive}, if the assignment satisfies $c$.
    If $c$ is inactive, it will be removed from the formula, while if it
    is active, it will be reduced to its core, because we have assigned values
    to all universal variables.

    For every $C_j$, we have $\Pr[C_j(\sigma) \text{ active}] \geq (1/2)^d$,
since for the clause to be active, it suffices for $\sigma$ to set all the (at
most $d$) universal literals to False. We thus have $\Pr[C_j(\sigma) \text{
inactive}] \leq (1 - (1/2)^d)$, and

    \begin{align*}
        \Pr[\text{all $C_j(\sigma)$ are inactive}]
            &\leq (1 - (1/2)^d)^{2^d d \ln k}\\
            &< e^{-d\ln k}\\
	    &= \frac{1}{k^d}, \end{align*} where we substituted $X = 2^d d \ln
k$ and used the fact that if clauses share no variables, then the event of one 
being active is independent from the others.  

We now apply this inequality by a union bound to all groups
$\mathcal{S}^\psi_C$.  If a group has a core of size $d$, then the
whole clause consists of existential variables. In this case, the set family
$\mathcal{S}^\psi_C$ consists of the empty clause. Thus, the number of
non-trivial groups is at most $\binom{k}{d-1} \leq k^{d-1}$, because each core
contains at most $d-1$ of the $k$ existential variables.  Using this bound, we
obtain

    \begin{align*}
        &\hphantom{\leq} \Pr\left[\bigvee_{C \subseteq \{ x, \neg x \mid x \in
        \mathbf{x} \}, \mathcal{S}^{\psi}_C \neq \{\emptyset\}} \text{max.\
        disjoint set of $\mathcal{S}^{\psi}_C$ has only inactive members}\right]\\
        &\leq \sum_{C \subseteq \{ x, \neg x \mid x \in \mathbf{x} \},
        \mathcal{S}^{\psi}_C \neq \{\emptyset\}} \Pr[\text{max.\ disjoint set of
        $\mathcal{S}^{\psi}_C$ has only inactive members}]\\
        &< k^{d-1} \cdot \left(\frac{1}{k^d}\right) = \frac{1}{k}.
    \end{align*}
    Hence, by the probabilistic method, there exists a strategy for the
    universal player that reduces all clauses of $\psi$ to their core.
\end{proof}

\subsection{Recursive Step and Analysis}

In the recursion step of the algorithm, there exists a group $G^\psi_C$ with set
family $\mathcal{S}^\psi_C$ that contains a hitting set $H$ of size $|H| \leq
dX$, in the sense that every clause in $\mathcal{S}^\psi_C$ contains a
universal variable of $H$. This follows because if we are not in the base case,
the maximal set of disjoint clauses we constructed for $\mathcal{S}^\psi_C$ has
size at most $X$ and intersects all other clauses of $\mathcal{S}^\psi_C$. 

The algorithm now enumerates all assignments $\sigma: H \to \{ 0, 1 \}$ and
recursively calls  itself on each $\psi(\sigma)$.

\begin{lemma}\label{lem:leaves} The recursion tree of the algorithm above has
at most $k^{O_d(k^{d-1})}$ leaves.  \end{lemma}

\begin{proof}

As our measure of progress, we define the weight of a formula $\psi$ by the
maximum number of universal variables in the clauses for each group:

    \[ w(\psi) = \sum_{C \subseteq \{ x, \neg x \mid x \in \mathbf{x} \}
        } \max_{V \in \mathcal{S}^\psi_C}
        |V|. \]

Before running the algorithm  all clauses have arity at most $d$. The number of
groups with non-zero weight is at most $k^{d-1}$, because in such a group all
clauses contain at least one universal variable, so at most $d-1$ of the $k$
existential variables.  Therefore, initially $w(\psi) \leq d\cdot k^{d-1}$.

Let $H$ be a hitting set of $\mathcal{S}^\psi_C$ and $\sigma: H\to\{0, 1\}$
be any assignment. Then we have $w(\psi(\sigma)) < w(\psi)$. Clearly, the arity
of all clauses of a group decreases after assigning a boolean value to each
variable in the hitting set.

Overall, in the recursion tree we branch with degree at most $2^{|H|} \leq 2^{dX}$
and the tree's depth is at most $w(\psi)$, so the number of leaves of the tree
is at most $(2^{dX})^{w(\psi)} \leq 2^{d^2 X(k,d) k^{d-1}} =
k^{O_d(k^{d-1})}$.
\end{proof}

Putting everything together we obtain the proof of the theorem of this section.

\begin{proof}[Proof of \Cref{thm:algo}.]

The algorithm proceeds as follows: First, we partition the clauses of $\psi$
into groups $G^\psi_C$ and compute the universal parts $\mathcal{S}^\psi_C$.
In each universal part (that does not simply consist of the empty clause) we
greedily construct a set of clauses whose variables are disjoint by starting
with an empty set and then arbitrarily selecting a clause to add to our set as
long as this is possible. If at some point our set contains at least $X$
clauses, we move on to the next group $\mathcal{S}^\psi_C$. This process can be
executed in polynomial time and can terminate in two ways:

\begin{enumerate}

\item The algorithm succeeds in constructing a set of $X$ disjoint clauses for
each group $\mathcal{S}^\psi_C$. In this case we construct $\psi|_\mathbf{x}$
and check if it is satisfiable (by going through all $2^k$ assignments). If it
is, we return True, otherwise we return False. Correctness of this step is
justified by Lemma~\ref{lem:reduce}.

\item The algorithm finds a group $\mathcal{S}^\psi_C$ and a set of fewer than
$X$ clauses of $\mathcal{S}^\psi_C$ such that all other clauses of
$\mathcal{S}^\psi_C$ share a universal variable with a clause of the set. Then,
the at most $dX$ variables of the clauses of the set form a hitting set $H$, in
the sense that all clauses of $\mathcal{S}^\psi_C$ contain at least one
variable of $H$. The algorithm branches into all assignments to the variables
of $H$ and returns False if and only if one branch returns False. This
branching is correct and produces a number of leaves bounded by
Lemma~\ref{lem:leaves}.

\end{enumerate}

Putting everything together, the running time is dominated by the number of
leaves of the recursion tree, which gives the stated running time.  \end{proof}

\section{Conclusions and Open Problems}

The main contribution of this paper is to resolve the open problem of
\cite{ErikssonLOOPR24} by showing that even for QBFs of arity $4$ the
double-exponential dependence of the running time on $k$ is necessary. This
leaves open as an intriguing open question the case of QBFs of arity $3$. We
conjecture that it should be possible to extend our lower bound to this case.

Furthermore, we showed that the complexity of $d$-QBF parameterized by $k$ does
become significantly better if we only allow two quantifier blocks
($\forall\exists$-QBF), as in this case we have an algorithm with parameter
dependence exponential in $k^{d-1}$, rather than double-exponential in $k$. On
the other hand, our double-exponential lower bound requires a small but
unbounded number of quantifier alternations to work (roughly $\log\log n$, as
each step of our construction replaces a formula of size $m$ with one of size
$O(\sqrt{m})$). This leaves open the question of what happens for a small but
constant number of alternations. Namely, is the complexity of solving
$\forall\exists\forall\exists$-QBF already double-exponential in $k$?

\bibliography{qbf}

\end{document}